\documentclass[12pt,peerreview]{IEEEtran}

\usepackage{amsthm}
\usepackage{amsmath}
\usepackage{amssymb}
\usepackage[colorlinks = true, linkcolor = black, citecolor = black, final]{hyperref}

\usepackage{comment}
\usepackage{cite}

\usepackage{tikz}
\usetikzlibrary{patterns,arrows,automata,shapes.symbols,positioning,calc,fit,shapes.multipart,chains}
\usepackage{graphicx}
\usepackage{multicol}
\usepackage{ marvosym }
\usepackage{wasysym}

\ifCLASSOPTIONcompsoc
  \usepackage[caption=false,font=normalsize,labelfont=sf,textfont=sf]{subfig}
\else
  \usepackage[caption=false,font=footnotesize]{subfig}
\fi
\usepackage{graphicx}
\usepackage{ulem}
\usepackage{caption}

\newcommand{\ds}{\displaystyle}

\newtheorem{Definition}{Definition}
\newtheorem{Lemma}{\bf Lemma}

\newtheorem{Theorem}{\bf Theorem}

\newtheorem{Remark}{Remark}

\def\Pr{{\rm \mathbf {Pr}}}
\def\E{{\rm \mathbf  E}}

\newcommand{\V}{\mathcal{V}}

\newcommand{\R}{\mathbb{R}}
\newcommand{\C}{\mathbb{C}}

\newcommand{\tDelta}{\tilde{\Delta}}
\newcommand{\hDelta}{\hat{\Delta}}
\newcommand{\src}{\theta}

\usepackage{epic}
\usepackage{float}
\usepackage{graphics}

\newcommand{\vtx}[3]{
\begin{picture}(2.5,0.8)
  \put(0,0){\vector(1,0){0.95}}
  \put(0.25,0.1){\makebox(0,0)[b]{#1}}
 
  \put(1.25,0){\circle{0.6}}
  \put(1.25,0){\makebox(0,0){#2}}
 
  \put(1.55,0){\line(1,0){0.95}}
  \put(2.25,0.1){\makebox(0,0)[b]{#3}}
\end{picture} }

\newcommand{\vtxud}[3]{
\begin{picture}(2.5,1.5)
  \put(0,0){\vector(1,1){0.95}}
  \put(0.35,0.45){\makebox(0,0)[br]{#1}}
 
  \put(1.25,1){\circle{0.6}}
  \put(1.25,1){\makebox(0,0){#2}}
 
  \put(1.55,0.95){\vector(1,-1){0.95}}
  \put(2.15,0.45){\makebox(0,0)[bl]{#3}}
\end{picture} }

\newcommand{\vtxudp}[3]{
\begin{picture}(2.5,1.5)
  \put(0,0){\vector(1,1){0.95}}
  \put(0.45,0.45){\makebox(0,0)[tl]{#1}}
 
  \put(1.25,1){\circle{0.6}}
  \put(1.25,1){\makebox(0,0){#2}}
 
  \put(1.55,0.95){\vector(1,-1){0.95}}
  \put(2.15,0.45){\makebox(0,0)[bl]{#3}}
\end{picture} }

\newcommand{\vtxudq}[3]{
\begin{picture}(2.5,1.5)
  \put(0,0){\vector(1,1){0.95}}
  \put(0.35,0.45){\makebox(0,0)[br]{#1}}
 
  \put(1.25,1){\circle{0.6}}
  \put(1.25,1){\makebox(0,0){#2}}
 
  \put(1.55,0.95){\vector(1,-1){0.95}}
  \put(2.15,0.45){\makebox(0,0)[tr]{#3}}
\end{picture} }

\newcommand{\vtxe}[2]{
\begin{picture}(1.6,0.8)
  \put(0,0){\vector(1,0){0.95}}
  \put(0.25,0.1){\makebox(0,0)[b]{#1}}
 
  \put(1.25,0){\circle{0.6}}
  \put(1.25,0){\makebox(0,0){#2}}
\end{picture} }

\begin{document}

\title{Age Distribution in Arbitrary Preemptive Memoryless Networks} 
\author{%
Rajai Nasser, Ibrahim Issa, and Ibrahim Abou-Faycal
\thanks{This paper is submitted in part to the IEEE International Symposium on Information Theory (ISIT) 2022.}
}

\maketitle

\normalem

\begin{abstract}
    We study the probability distribution of age of information (AoI) in arbitrary networks with memoryless service times. A source node generates packets following a Poisson process, and then the packets are forwarded across the network in such a way that newer updates preempt older ones. This model is equivalent to gossip networks that was recently studied by Yates, and for which he obtained a recursive formula allowing the computation for the average AoI. In this paper, we obtain a very simple characterization of the stationary distribution of AoI at every node in the network. This allows for the computation of the average of an arbitrary function of the age. In particular, we can compute age-violation probabilities. Furthermore, we show how it is possible to use insights from our simple characterization in order to substantially reduce the computation time of average AoIs in some structured networks. Finally, we describe how it is possible to use our characterization in order to obtain faster and more accurate Monte Carlo simulations estimating the average AoI, or the average of an arbitrary function of the age.
\end{abstract}

\section{Introduction}


Many new technologies (e.g., vehicular networks, sensor networks, IoT applications, etc \ldots) rely on the assumption that the information that a node in a network has about other nodes is as fresh as possible. Timeliness of updates thus emerged as a new research topic in the study of networks. While optimizing for utilization and/or delay (latency) can be correlated with getting better timeliness at the receiver, such strategies do not necessarily optimize timeliness, even in very simple settings \cite{KaulYatesGruteser-2012Infocom}.

Age of information (AoI) \cite{KaulYatesGruteser-Globcom2011} is a metric that better captures the concept of updates' timeliness. Consider a network node that monitors the status of another network node. The monitor will continuously receive updates from the monitored node -- henceforth referred to as transmitter -- and the goal is to keep the information that the monitor has about the transmitter as fresh as possible. For every instant of time $t\geq 0$, let $g(t)$ be the timestamp of the most recent update that the monitor has (successfully) received from the transmitter. The instantaneous age of information at the monitor at time $t$ is defined as
$$\Delta(t)=t-g(t)\,,$$
and the average age of information is defined as
$$\Delta = \lim_{\tau\to\infty} \frac{1}{\tau} \int_0^\tau \Delta(t) \, dt\,.$$

In Fig.~\ref{fig:fig_instantaneous_age_general}, we show an example illustrating how the instantaneous age varies with time. In this figure, $t_i'$ represents the instant at which the $i$-th successfully received message was decoded at the receiver, and $t_i=g(t_i')$ represents the generation time of this message at the source.

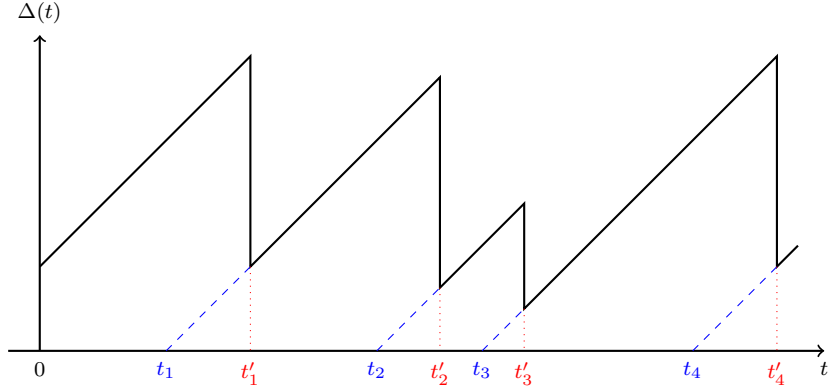
\begin{figure}[ht]
\centering
	\begin{tikzpicture}[scale=0.7,font=\scriptsize]
		\draw[thick,->] (-2,0) -- (13.5,0) node[anchor=north] {$t$};

		\draw[color=blue] 	(1,0) node [anchor=north] {$t_1$};
		\draw[color=red] 	(2.6,0) node [anchor=north] {$t_1'$};
		\draw[color=blue] 	(5,0) node [anchor=north] {$t_2$};
		\draw[color=red] 	(6.2,0) node [anchor=north] {$t_2'$};
		\draw[color=blue] 	(7,0) node [anchor=north] {$t_3$};
		\draw[color=red] 	(7.8,0) node [anchor=north] {$t_3'$};
		\draw[color=blue] 	(11,0) node [anchor=north] {$t_4$};
		\draw[color=red] 	(12.6,0) node [anchor=north] {$t_4'$};
		
		\draw[thick,->] (-1.4,0) node[anchor=north] {$0$}-- (-1.4,6) node[anchor=south] {$\Delta(t)$};
		
		\draw[thick] (-1.4,1.6) -- (2.6,5.6) -- (2.6,1.6) -- (6.2,5.2) -- (6.2,1.2) -- (7.8,2.8) -- (7.8,0.8) --
		(12.6,5.6) -- (12.6,1.6) -- (13,2);
		
		\draw[dashed,color=blue] (1,0) -- (2.6,1.6);
		\draw[dashed,color=blue] (5,0) -- (6.2,1.2);
		\draw[dashed,color=blue] (7,0) -- (7.8,0.8);
		\draw[dashed,color=blue] (11,0) -- (12.6,1.6);

		\draw[dotted,color=red] (2.6,1.6) -- (2.6,0);
		\draw[dotted,color=red] (6.2,1.2) -- (6.2,0);
		\draw[dotted,color=red] (7.8,0.8) -- (7.8,0);
		\draw[dotted,color=red] (12.6,1.6) --(12.6,0);
	\end{tikzpicture}
	\caption{The instantaneous age $\Delta(t)$.}
	\label{fig:fig_instantaneous_age_general}
\end{figure}

A large number of papers studied the AoI from a queuing-theoretic perspective. Packets are generated at the transmitter with independent (random) interarrival times, and these packets (i.e., updates) are transmitted to the monitor through one link\footnote{In the queuing theory terminology, the link can be thought of as a server.}. The time needed for the packet to be successfully received at the monitor is called the service time, and is modelled as a random variable. By specifying the probability distributions of packet interarrival time and of the service time, we obtain different queuing-theoretic models (e.g., M/M/1, G/M/1, G/G/1, etc \ldots). The AoI was studied under these queuing models for various transmission scheduling policies (FCFS, LCFS, LCFS with preemption in service, LCFS with preemption in the waiting queue, etc \ldots). See e.g., \cite{KaulYatesGruteser-2012Infocom,2012CISS-KaulYatesGruteser,CostaCodreanuEphremides2014ISIT,KamKompellaEphremides2013ISIT,NajmNasser-ISIT2016,YatesKaul-2012ISIT,InoueEtAl,SoysalUlukusGG11,YatesKaulTrans2019,NajmTelatar2018}.

The age of information problem has also been studied under resource allocation constraints, such as energy in~\cite{Elif-2015ITA,Yates-2015ISIT,ArafaUlukus17,BacinogluUysal-ISIT17,BacinogluSunUysalMutlu-ISIT18,WuYangWu18,FengYangTransComm2021,ArafaYangUlukusPoorTrans2020,ArafaYangUlukusPoorAllerton2018,ArafaYangUlukusPoorISIT2019,LiuEtAlTransIt2016}. Other staleness metrics that have been studied include peak age \cite{CostaCodreanuEphremides2014ISIT,CostaCodreanuEphremides2016,InoueEtAl} and age-violation probability \cite{BedewySunShroffTrans2019}, among others. For an excellent recent survey about AoI, see \cite{YatesEtAlJSACSurvey}.

In this paper, we study the age of information problem in a networking setting. In \cite{BedewySunShroff17}, Bedewy \emph{et al.} showed that if a source transmits status updates in an arbitrary multihop network where all links have exponentially distributed service times, then the LCFS policy with preemption in service is age optimal. In \cite{YatesPreempServs2018}, Yates used the stochastic hybrid systems (SHS) formalism \cite{YatesKaulTrans2019,TeelSubbaramanSferlazza2014,Hespanha2006} to derive an exact formula for the average AoI in a \emph{line network} following the LCFS policy with preemption in service. The stationary distribution of AoI in line networks was later derived by Yates in \cite{YatesAoIMomTrans2018} by studying its moment-generating function (MGF). In \cite{YatesGossipISIT2021,YatesGossipIWSPAWC2021}, Yates extended the result of \cite{YatesPreempServs2018} to arbitrary multihop (gossip) networks.

Yates predicted in \cite{YatesGossipIWSPAWC2021} that the techniques of \cite{YatesAoIMomTrans2018}, namely MGF, can be used to derive distributional properties of the age of information in arbitrary gossip networks. In this paper, we confirm Yates' prediction by deriving the stationary distribution of the age of information in arbitrary networks where all nodes follow a preemptive policy in service. While we do use the moment-generating function to achieve this, our approach departs from that of \cite{YatesAoIMomTrans2018} in that we do not use the stochastic hybrid systems formalism. 

Since we are able to exactly compute the AoI distribution, our results allow for the computation of many other AoI-related staleness metrics. For example, we can compute the average of an arbitrary function of the age. In particular, we can compute the age-violation probability.

It is worth mentioning that the characterization that we get for the AoI distribution has a very simple form. The simplicity of this characterization can be leveraged to obtain faster and more accurate Monte Carlo simulations estimating the average AoI (or the average of any function of the age). Finally, if the network is structured, we can leverage this structure in addition to the simplicity of our characterization in order to reduce the computation time of the AoI distribution or the average AoI.

\section{Main Result}

\label{sec:MainResult}

\begin{Definition}[Single-Source Network] \label{def:graph}
Let $G=(V,E)$ be a weighted directed graph where $V$ is the set of vertices and $E$ is the set of edges. We say that $G$ is a Single-Source Network (SSN) if:
\begin{itemize}
    \item It has a unique node with in-degree zero. We call it \emph{source} and denote it by $\src$.
    \item All nodes are reachable from the source $\src$.
    \item It has no self loops.
\end{itemize}
\end{Definition}

\paragraph{{\bf Model}} 
To model packet transmission through an SSN represented by a graph $G$, we assume that the source node generates packets according to a Poisson process of rate $\lambda$. Each node $v\in V$ represents a buffer of capacity 1, and each edge $e \in E$ represents a queue with exponentially distributed service of rate $\mu_e$ (i.e., ./M/1 queue), where $\mu_e$ is the corresponding weight of the edge $e$. For every node $v \in V$ and every instant of time $t \geq 0$, let $g_v(t)$ denote the generation time of the freshest packet that the node $v$ has received. Define the age process $\Delta_v(t)$ at node $v$ in the standard way: \[ \Delta_v(t) = t- g_v(t)\,. \]
Each node transmits the packet in its buffer (if it exists) through all its outgoing edges\footnote{The node aims to deliver the packet through \emph{all} connected servers. We may assume that the packet, even after successful transmission, remains in the buffer until it is explicitly preempted.}, and implements a preemption policy in service. More precisely, if the service corresponding to an edge $(u,v) \in E$ terminates, then the node $v$ compares $g_v(t)$ with $g_u(t)$ (which is the generation time of the packet just received); if $g_u(t) > g_v(t)$, the received packet is newer than the existing one, hence $v$ preempts its own packet and starts transmitting the new packet instead; otherwise (i.e., if $g_u(t) \leq g_v(t)$), the received packet is ignored. Finally, note that the (random) service times are mutually independent.

\paragraph{{\bf Notation}}
The notation $S \sim \mathrm{Exp}(\mu)$ is used to indicate that $S$ is a random variable with exponential distribution of rate $\mu\geq 0$, i.e., $S$ is a continuous random variable with pdf given by $p_{S}(t) = \mu e^{-\mu t}$, $t \geq 0$. For $u, v \in V$, we write $\mathcal{P}(u \to v)$ to denote the set of (directed) paths in $G$ from $u$ to $v$.


Our main theorem provides a computable characterization of the average age of information by deriving the stationary distribution of each age process.

\begin{Theorem} \label{thm:main}
Consider an SSN graph $G=(V,E)$ representing a single-source network. For each $v \in V$, let $\Delta_v(t)$ represent the age process at node $v$ at time $t$. Define the random variables $(\tilde{\Delta}_v)_{v \in V}$ as follows:
\begin{align}
    \tilde{\Delta}_{\src} & \sim \mathrm{Exp}(\lambda)\,, \label{eq:thm-main-delta-s} \\
    \tilde{\Delta}_v & = \tilde{\Delta}_{\src}+ \min_{P \in \mathcal{P}(\src \rightarrow v)} \sum_{e \in P} S_e,\quad \forall v\neq {\src}\,, \label{eq:thm-main-delta-v} 
\end{align}
where $S_{e} \sim \mathrm{Exp}(\mu_{e})$, and the random variables $\{\tDelta_{\src},(S_{e})_{e \in E}\}$ are mutually independent. 
For each $v \in V$, let $\Delta_v$ be a random variable distributed according to the stationary distribution\footnote{It is worth noting that the AoIs process $\big((\Delta_v(t))_{v\in V}\big)_{t\geq 0}$ is ergodic, and hence they have a stationary distribution.} of $\Delta_v(t)$. Then 
\begin{equation}
\Delta_v = \tilde{\Delta}_v\text{ in distribution.}\label{eq:thm-main-eq-dist} 
\end{equation}
In particular, the average age at the destination is given by
\begin{align} \label{eq:thm-main-age}
    \E[\Delta_d] = \frac{1}{\lambda} + \E \left[ \min_{P \in \mathcal{P}({\src} \rightarrow d)} \sum_{e \in P} S_e \right].
\end{align}
\end{Theorem}

\begin{Remark}
We emphasize that even though the marginal distributions corresponding to individual random variables in $(\Delta_v)_{v\in V}$ agree with the marginal distributions of the corresponding random variables in $(\tilde{\Delta}_v)_{v \in V}$, the random variables $(\Delta_v)_{v\in V}$ and $(\tilde{\Delta}_v)_{v \in V}$ do not have the same joint distribution. In order to see why, consider two vertices $u,v$ that are connected by an edge $(u,v)\in E$. Imagine that at some time $t_0$ the node $u$ succeeds in delivering a new update to $v$ so that $\Delta_v(t_0)=\Delta_u(t_0)$. Shortly after $t_0$, both $\Delta_v(t)$ and $\Delta_u(t)$ will increase linearly, and hence we will have $\Delta_v(t)=\Delta_u(t)$, until either $u$ or $v$ receives a new update. Since this can happen with nonzero probability, we can see that we have $\Delta_v(t)=\Delta_u(t)$ on a nonzero fraction of the (positive) real line $t\in\R^+$, which means that in the stationary distribution we must have $\Delta_u=\Delta_v$ with nonzero probability. On the other hand, from the definition of $\tilde{\Delta}_u$ and $\tilde{\Delta}_v$, we can easily see that $\tilde{\Delta}_u=\tilde{\Delta}_v$ happens with zero probability.
\end{Remark}

Before proving Theorem \ref{thm:main}, it will be notationally convenient to augment the graph $G$ by adding a virtual node $\src'$ to $V$ and an edge $(\src',\src)$ to $E$. More precisely, let
\begin{align} \label{eq:augment-graph}
    V'=V\cup\{\src'\}\ \text{ and } E'=E\cup\{(\src',\src)\}\,,
\end{align}
and associate to the edge $(\src',\src)$ the rate $$\mu_{\src'\src}:=\lambda\,$$ and the random variable $$S_{\src'\src} \sim \mathrm{Exp}(\lambda)\,.$$ Let $G'=(V',E')$ be the augmented graph.
For the new node $\src'$, let $$\Delta_{\src'}(t)=0\text{ for all }t \geq 0\,,$$ and assume that $\src'$ continuously transmits new packets, i.e., once a packet is received from $\src'$ to $\src$, $\src'$ generates and starts transmitting a new packet. Since the associated service rate of the edge $(\src,\src')$ is exponential with rate $\lambda$, this is equivalent to $\src$ producing packets according to a Poisson process of rate $\lambda$. In this context, Equations~\eqref{eq:thm-main-delta-s} and~\eqref{eq:thm-main-delta-v} can be rewritten as
\begin{align}
    \tilde{\Delta}_{\src'} & = 0\,, \label{eq:thm-main-delta-s-2} \\
    \tilde{\Delta}_v & = \min_{P \in \mathcal{P}(\src' \rightarrow v)} \sum_{e \in P} S_e, \text{ for } v \neq \src'\,.\label{eq:thm-main-delta-v-2} 
\end{align}

For ease of notation, we will define
\begin{equation}
    \V'=\{A\subseteq V':\; A \neq \varnothing\} \text{ and } \V = \{A \subseteq \V:\; A \neq \varnothing \},
\end{equation}
to be the collections of nonempty subsets of $V'$ and $V$, respectively. Note that $\V \subseteq \V'$. Now for every $A \in \V'$, define 
$$\Delta_A = \min_{v \in A} \Delta_v\qquad\text{and}\qquad\tilde{\Delta}_A = \min_{v \in A} \tilde{\Delta}_v\,,$$
and for every $t\geq 0$, define
$$\Delta_A(t) = \min_{v \in A} \Delta_v(t)\,.$$
Note that $\Delta_A$ corresponds to the stationary distribution of the stochastic process $(\Delta_A(t))_{t\geq 0}$.

We will prove a stronger statement than that of Theorem~\ref{thm:main}:

\begin{Theorem} \label{thm:statDistSubsets}
For all $A \in \V'$, we have
\begin{align}
    \Delta_A = \tDelta_A \text{ in distribution}\,.
\end{align}
\end{Theorem}
In order to show Theorem \ref{thm:statDistSubsets}, we will prove that the moment-generating functions of $\Delta_A$ and $\tDelta_A$ are equal, i.e., for all $s \in \C$ for which $\E \left[ e^{s \Delta_A} \right]$ exists, we have
\begin{align} \label{eq:main-eq-MGF}
    \E \left[ e^{s \Delta_A} \right] =  \E \left[ e^{s \tDelta_A} \right]\,.
\end{align}
Now to prove \eqref{eq:main-eq-MGF}, we will first study the (joint) moment-generating function of the collection of random variables $(\Delta_A)_{A\in\V'}$. Namely, we will study the function defined as
\begin{align} \label{eq:def-MGF}
F\big((s_A)_{A\in \V' }\big) :=   \E \left[ \exp\left(\sum_{A\in \V' }s_A \Delta_A\right) \right]\,,
\end{align}
where we associate with each $A\in\V'$ a complex variable $s_A\in\C$.

%

\subsection{Description of the evolution of the age processes}
In order to study $F$, it will be useful to describe an equivalent formulation of $$(\Delta_A(t))_{A\in\V',~t\geq 0}\,,$$ which is particularly convenient for our analysis. 
Since $\displaystyle\Delta_A(t) = \min_{v \in A} \Delta_v(t)$,
it is sufficient to describe how we can get $\Delta_v(t)$ for $v\in V'$ and $t\geq 0$. For $v=\src'$, recall that
\begin{equation}
    \Delta_{\src'}(t) = 0\,,\quad \forall t\geq 0\,.
\end{equation}
In the following, we will describe how we can get $\Delta_v(t)$ for $v\in V$ and $t\geq 0$. The idea is that the ages grow linearly at each node, until the service at some edge $(u,v)$ terminates. At that point, the age at $v$ (and all sets $A$ containing $v$) may need to be updated. 

To make the previous description precise, let $(T_n)_{n\geq 1}$ be a sequence of i.i.d. random variables distributed as  $\displaystyle T \sim \min_{e \in E'} S_e$, i.e.,  $T\sim \mathrm{Exp}(\mu)$ where $\mu$ is given by 
\begin{align} \label{eq:total-rate}
  \mu = \sum_{e\in E'} \mu_e\,.  
\end{align}
This defines a Poisson process of rate $\mu$. Let $B_0 = 0$ and for every $n\geq 1$ define
\begin{equation}
B_n = \sum_{i=1}^n T_i\,.
\end{equation}
Note that $B_n$ is the time at which the $n$-th event of the Poisson process $(T_n)_{n\geq 1}$ occurs. That is, the service at some edge $e$ has terminated. To associate this event with a specific edge, we will (implicitly) apply a thinning operation to the Poisson process $(T_n)_{n\geq 1}$ in order to (implicitly) get $|E'|$ Poisson processes of rates $(\mu_e)_{e\in E'}$, respectively.
Independently of $(T_n)_{n\geq 1}$, generate a sequence of i.i.d. random variables $(E_n)_{n\geq 1}$ taking values in $E'$ in such a way that for every $e\in E'$, we have
\begin{equation}
    \Pr[E_n = e] = \frac{\mu_e}{\mu}\,,
\end{equation}
which corresponds to the probability that $\displaystyle S_e = \min_{e' \in E'} S_{e'}$.
If $E_n=e$, then at $t=B_n$, the service time corresponding to the edge $e$ has just finished. 

Given the processes $(T_n)_{n \geq 1}$ and $(E_n)_{n \geq 1}$, we are ready to describe the evolution of the age processes. 
For every $n\geq 1$, every $B_{n-1}\leq t<B_n$ (i.e., no events occurring), and every $v\in V$, we set
\begin{equation}
    \Delta_v(t) = \Delta_v(B_{n-1})+t-B_{n-1}\,.
\end{equation}
That is, all the ages $(\Delta_v(t))_{v\in V}$ increase linearly on every interval $[B_{n-1},B_n)$. 
It remains to describe how to get $(\Delta_v(B_n))_{v\in V}$ for $n\geq 0$. To that end, we set $(\Delta_v(B_0))_{v\in V}=(\Delta_v(0))_{v\in V}$ to arbitrary values\footnote{Due to the ergodicity of the system, the function $F$ in \eqref{eq:def-MGF} is insensitive to the initial values $(\Delta_v(0))_{v\in V}$.} and then for every $n\geq 1$, we define $(\Delta_v(B_n))_{v\in V}$ in terms of $(\Delta_v(B_{n-1}))_{v\in V}$ as follows: If $E_n=(u,w)$, then for every $v\in V$, we set
\begin{align}  
\Delta_v(B_n)=
\begin{cases}
\min\big\{\Delta_{u}(B_{n-1}),\Delta_{w}(B_{n-1})\big\}+T_n\,, & \text{if } v=w\neq\src\,, \\
0\,, & \text{if } v =w=\src\,,\\
\Delta_{v}(B_{n-1})+T_n\,, & \text{if } v \neq w\,.
\end{cases} \label{eq:updateGivenEnEquale}
\end{align}
Consequently, for every $A \in \V$, we have
\begin{align}
\Delta_A(B_n) = 
\begin{cases}
\Delta_{A \cup \{u\}} (B_{n-1}) + T_n\,, &  \text{if } \src\neq w \in A\,, \\
0\,, & \text{if } \src=  w \in A\,,\\
\Delta_{A} (B_{n-1}) + T_n\,, & \text{if } w \notin A\,.\\
\end{cases} \label{eq:updateGivenEnEqualeA} 
\end{align}


%

\subsection{Recursive relations for the MGF} We are now ready to study $F$ in~\eqref{eq:def-MGF}. Note that by ergodicity we may rewrite $F$ as 
\begin{align} \label{eq:MGF-erg}
F\left((s_A)_{A\in\V'}\right) =\lim_{\tau\to\infty}\frac{1}{\tau}\int_{0}^{\tau} \exp\left(\sum_{A\in\V'}s_A \Delta_A(t)\right) dt\,.
\end{align}
Roughly speaking, we first study $F$ at ``critical points'' at which events occur. To that end, consider the following lemma.
\begin{Lemma} \label{lem:MGF-Relations}
Define 
\begin{align} \label{eq:MGF-critical}
  f\left((s_A)_{A\in\V}\right) &= \lim_{n\to\infty} \E\left[\exp\left(\sum_{A\in\V}s_A \Delta_A(B_{n})\right)\right].
\end{align}
Then,
\begin{align} \label{eq:MGF-critical-to-MGF}
F\left((s_A)_{A\in\V'}\right) =\frac{\mu}{\mu - \resizebox{0.03\textwidth}{!}{$\displaystyle\sum_{A\in\V}$}~s_A}\cdot f\left((s_A)_{A\in\V}\right)\,.
\end{align}
\end{Lemma}
\begin{proof}
See Appendix \ref{app:proof-lem-MGF-Relations}.
\end{proof}

The update equation in~\eqref{eq:updateGivenEnEqualeA} will allow us to derive a recursive relations for $f$, and hence for $F$ as well. ``Marginalizing'' $F$ to compute $\E[\exp (s_A\Delta_A )]$ leads to analogous recursive relations:

\begin{Lemma}
\label{lem:MGF-Relations-recur}
Defining
\begin{align} \label{eq:lem-MGF-marg}
  F_{A'}\left(s \right) &:= F\left((s)_{A=A'}~,~(0)_{\substack{A\in\V':\\A\neq A'}}\right) =\E \left[ \exp\left(s \Delta_{A'}\right) \right]\,,\quad\text{for every }A' \in \V'\,,
\end{align}
we get following:
\begin{itemize}
\item If $A'\notin \V$, i.e., if $\src'\in A'$, then
\begin{align} \label{eq:lem-MGF-src-prime}
  F_{A'}\left(s\right) &= 1\,.
\end{align}
\item If $\src\in A'\in \V$, then
\begin{align} \label{eq:lem-MGF-src}
  F_{A'}\left(s\right) &= \frac{\mu_{\src'\src}}{\mu_{\src'\src}-s}=\frac{\lambda}{\lambda-s}\,.
\end{align}
\item If $\src\notin A'\in \V$, then
\begin{align} \label{eq:lem-MGF-marg-recur}
  F_{A'}\left(s\right) &= \frac{1}{\resizebox{0.1\textwidth}{!}{$\displaystyle\sum_{\substack{(u,v)\in E':\\v\in A'\text{ and }u\notin A'}}$} \mu_{uv} - s} \cdot \sum_{\substack{(u,v)\in E':\\v\in A'\text{ and }u\notin A'}} \mu_{uv}\cdot F_{A'\cup\{u\}}\left(s\right)\,.
\end{align}
\end{itemize}
\end{Lemma}
\begin{proof}
See Appendix \ref{app:proof-lem-MGF-Relations-recur}.
\end{proof}

\subsection{Proof of Theorem~\ref{thm:main}}

For every $A \in \V'$, let
\begin{align}
\tilde{F}_{A} (s) = \E \left[ \exp \left( s \tDelta_{A} \right) \right]
\end{align}
be the right-hand side of~\eqref{eq:main-eq-MGF}. We need to show that
\begin{align} \label{eq:EquationMain}
\tilde{F}_{A} (s) = F_{A} (s)
\end{align}
for all $A \in \V'$ and all $s \in \mathbb{C}$ for which $F_A(s)$ exists.

Note that for every $A$ satisfying $\src' \in A$, we have $\Delta_A = \tDelta_A = 0$ by definition, hence 
\begin{equation}
F_A(s) = \tilde{F}_A(s) = 1\,.
\label{eq:lem-F-tF-src-prime}
\end{equation}
It remains to show the equality for $A \in \V$. Note that if $\src\in A\in \V$, then \eqref{eq:lem-MGF-src} implies that
\begin{equation}
F_A(s)=\frac{\mu_{\src'\src}}{\mu_{\src'\src}-s}= \E\left[sS_{\src'\src}\right]= \E\left[s\tDelta_{\src}\right]\stackrel{\text{(a)}}= \E\left[s\tDelta_{A}\right]=\tilde{F}_A(s)\,,
\label{eq:lem-F-tF-src}
\end{equation}
where (a) follows from the fact that for all $v\in V$, we have $\tDelta_v\geq\tDelta_\src$, and hence $\tDelta_{A}=\tDelta_{\src}$ for $\src\in A\in \V$.

We will now proceed by reverse induction on the size of $A$ in order to show that $F_A(s)=\tilde{F}_A(s)$ for all $A\in \V'$.  Consider $|A|=|V'|$, i.e., $A = V'$, then from \eqref{eq:lem-F-tF-src-prime} we have
$F_{V'}(s) = \tilde{F}_{V'}(s)=1$.

Now let $V'\neq A \in\V'$ and assume that \eqref{eq:EquationMain} holds for every $A'$ satisfying $|A'|>|A|$. If $\src' \in A$ or $\src\in A$, then the equality follows from Equations \eqref{eq:lem-F-tF-src-prime} and \eqref{eq:lem-F-tF-src}. If $\src\notin A\in \V$, then using Equation~\eqref{eq:lem-MGF-marg-recur} of Lemma~\ref{lem:MGF-Relations-recur} and applying the induction hypothesis on $A\cup\{u\}$ for $u\notin A$, we get
\begin{align}
F_{A}\left(s \right) 
& = \frac{1}{\resizebox{0.1\textwidth}{!}{$\displaystyle\sum_{\substack{(u,v)\in E':\\v\in A\text{ and }u\notin A}}$} \mu_{uv} - s} \cdot \sum_{\substack{(u,v)\in E':\\v\in A\text{ and }u\notin A}} \mu_{uv}\cdot \tilde{F}_{A\cup\{u\}}\left(s\right)\,.
\end{align}
It remains to show that
\begin{align} \label{eq:EquationToShow}
\tilde{F}_{A}\left(s \right)  = \frac{1}{\resizebox{0.1\textwidth}{!}{$\displaystyle\sum_{\substack{(u,v)\in E':\\v\in A\text{ and }u\notin A}}$} \mu_{uv} - s} \cdot \sum_{\substack{(u,v)\in E':\\v\in A\text{ and }u\notin A}} \mu_{uv}\cdot \tilde{F}_{A\cup\{u\}}\left(s\right)\,.
\end{align}
In order to simplify the above equation, let us introduce
\begin{equation}
    E_A := \{ (u,v) \in E':  u \notin A,~ v \in A\}\,,\label{eq:def-E_A}
\end{equation}
and
\begin{equation}
    \mu_A := \sum_{(u,v) \in E_A} \mu_{uv}\,.\label{eq:def-mu_A}
\end{equation}
Rewriting~\eqref{eq:EquationToShow}, we need to show that
\begin{align}
\label{eq:EquationToShow-2}
\tilde{F}_{A}\left(s \right)  = \frac{1}{\mu_A - s} \cdot \sum_{(u,v) \in E_A } \mu_{uv}\cdot \tilde{F}_{A\cup\{u\}}\left(s\right).
\end{align}
Let 
\begin{align}
S_A = \min_{(u,v) \in E_A} S_{uv}.
\end{align} Then $S_A$ is exponentially distributed with rate $\mu_A$. Now note that
\begin{align}
\tilde{F}_{A}\left(s \right)  & = \E \left[ \exp \left\lbrace s \tDelta_A \right\rbrace \right] \\
& =  \sum_{(u,v) \in E_A} \Pr \left[S_A = S_{uv} \right]\cdot  \E \left[ \exp \left\lbrace s \tDelta_A \right\rbrace  \middle| S_A =S_{uv} \right] \\
& = \sum_{(u,v) \in E_A} \frac{\mu_{uv}}{\mu_A} \cdot \E \left[ \exp \left\lbrace s \tDelta_A \right\rbrace  \middle| S_A =S_{uv} \right]. \label{eq:conditioning}
\end{align}
To get~\eqref{eq:EquationToShow-2}, it is sufficient to show that for every $(u^\star,v^\star) \in E_A$, we have
\begin{align}
\label{eq:EquationToShow-3}
\E \left[ \exp \left\lbrace s \tDelta_A \right\rbrace  \middle| S_A =S_{u^\star v^\star} \right] =\frac{\mu_A}{\mu_A-s} \cdot\tilde{F}_{A \cup \{u^\star \}}(s).
\end{align}

To that end, the following lemma will be useful.
\begin{Lemma} \label{lem:age-by-path}
Given $A \in \V$, define
\begin{align}
V'_A & := \{ u \in V': \exists v \in A \text{ such that } (u,v) \in E_A \}, \\
\text{ and } 
V^{'\star}_A & := \{u \in V'_A: \exists P \in \mathcal{P}(\src' \to u) \text{ such that all nodes } w \text{ in } P \text{ satisfy } w \notin A \}.
\end{align}
Moreover, given $u , v \in V'\setminus  A$, define
\begin{align}
\mathcal{P}_{A^c} (u \to v) = \{P \in \mathcal{P}(u \to v): \text{all nodes } w \text{ in } P \text{ satisfy } w \notin A \}.
\end{align}
Then,
\begin{align}
\tDelta_A & = \min_{u \in V^{'}_A } \left\lbrace \min_{P \in \mathcal{P} (\src' \to u)} \sum_{e \in P} S_e + \min_{v:~(u,v) \in E_A} S_{uv} \right\rbrace. \label{eq:lem-age-by-path-0}\\
& = \min_{u \in V^{'\star}_A } \left\lbrace \min_{P \in \mathcal{P}_{A^c} (\src' \to u)} \sum_{e \in P} S_e + \min_{v:~(u,v) \in E_A} S_{uv} \right\rbrace  \label{eq:lem-age-by-path}\\
& = \min_{u \in V^{'\star}_A } \left\lbrace \tDelta_{u,A^c}+ S_{u,A} \right\rbrace\,,\label{eq:lem-age-by-path-simple}
\end{align}
where
\begin{equation}
    \tDelta_{u,A^c} := \min_{P \in \mathcal{P}_{A^c} (\src' \to u)} \sum_{e \in P} S_e \quad\text{for every }u\in  V^{'\star}_A\,,
    \label{eq:tdelta-Ac}
\end{equation}
and
\begin{equation}
    S_{u,A}:=\min_{v:\; (u,v) \in E_A} S_{uv}\quad\text{for every }u\in  V^{'\star}_A\,.
\end{equation}
\end{Lemma}
The proof of the lemma is straightforward: The paths described in~\eqref{eq:lem-age-by-path-0} cover all possible paths from $\src'$ to $A$ that end with an edge in $E_A$. Any other path would have a larger weight. The paths in~\eqref{eq:lem-age-by-path} are exactly the paths from $\src'$ to $A$ that have only one node in $A$, namely the final node. Any other path would have a larger weight.

The lemma is stating the following intuitive fact: The shortest path from $\src'$ to $A$ has only one node in $A$, namely the final node. Note that $E_A$ is non-empty because $\src' \notin A$.

From Lemma~\ref{lem:age-by-path}, we get
\begin{align}
\tDelta_A & = \min \left\lbrace \tDelta_{u^\star,A^c} + S_{u^\ast,A}~,~ \min_{ \substack{u \in V^{'\star}_A: \\ u \neq u^\star} } \left\lbrace \tDelta_{u,A^c} + S_{u,A} \right\rbrace \right\rbrace\,.
\end{align}

Now conditioned on $S_{u^\star v^\star} = S_A$, we have $S_{u^\star,A}=S_{u^\star v^\star}$, and we can rewrite
\begin{align} \label{eq:tdelta-rewritten}
\tDelta_A & = S_{u^\star v^\star} + \min \left\lbrace \tDelta_{u^\star,A^c}~,~ \min_{ \substack{u \in V^{'\star}_A: \\ u \neq u^\star} } \left\lbrace \tDelta_{u,A^c} + S_{u,A} -  S_{u^\star v^\star}\right\rbrace \right\rbrace.
\end{align}
Note that the paths that appear in~\eqref{eq:tdelta-Ac} do not contain any edge in $E_A$, since by definition the paths are in $A^c$ and every edge in $E_A$ contains one element in $A$. Hence, these terms are independent from $S_{u^\star v^\star}$. Therefore, $(\tDelta_{u,A^c})_{u\in V^{'\star}_A}$ is independent from $S_{u^\star v^\star}$. Finally, note that for $u\in V^{'\star}_A\setminus\{u^\star\}$, $S_{u,A}$ is independent from $S_{u^\star v^\star}$, and hence by the memorylessness property of exponential random variables\footnote{Note that $S_{u,A}$ is an exponential random variable with rate $\displaystyle\sum_{v: (u,v) \in E_A} \mu_{uv}$.}, given $S_{u^\star v^\star}< S_{u,A}$, the random variable $S_{u,A}-S_{u^\star v^\star}$ is conditionally independent from $S_{u^\star v^\star}$, and its conditional distribution is the same as the unconditional distribution of $S_{u,A}$. Hence, given $S_{u^\star v^\star} = S_{u^\star,A}= S_A$, the second term of the right-hand side of~\eqref{eq:tdelta-rewritten} is conditionally independent from $S_{u^\star v^\star}$ and its conditional distribution is the same as the unconditional distribution of
\begin{align}
\min \left\lbrace \tDelta_{u^\star,A^c}~,~ \min_{ \substack{u \in V^{'\star}_A: \\ u \neq u^\star} } \left\lbrace \tDelta_{u,A^c} + S_{u,A}\right\rbrace \right\rbrace\,.
\end{align}
Therefore,
\begin{align}
& \E \left[ \exp \left\lbrace s \tDelta_A \right\rbrace  \middle| S_A =S_{u^\star v^\star} \right] \\
& = \E \left[ \exp \left\lbrace  s\cdot S_{u^\star v^\star}\right\rbrace \middle| S_A =S_{u^\star v^\star} \right] \cdot 
\E \left[ \exp  \left( s \cdot 
\min \left\lbrace \tDelta_{u^\star,A^c}~,~ \min_{ \substack{u \in V^{'\star}_A: \\ u \neq u^\star} } \left\lbrace \tDelta_{u,A^c} + S_{u,A}\right\rbrace \right\rbrace \right ) \right] \\
& = \frac{\mu_A}{\mu_A -s }\cdot \E \left[ \exp  \left( s \cdot 
\min \left\lbrace \tDelta_{u^\star,A^c}~,~ \min_{ \substack{u \in V^{'\star}_A: \\ u \neq u^\star} } \left\lbrace \tDelta_{u,A^c} + S_{u,A}\right\rbrace \right\rbrace \right ) \right], \label{eq:almost-there}
\end{align}
where the first equality follows from the argument above. Finally, note that
\begin{align}
\min \left\lbrace \tDelta_{u^\star,A^c}~,~ \min_{ \substack{u \in V^{'\star}_A: \\ u \neq u^\star} } \left\lbrace \tDelta_{u,A^c} + S_{u,A}\right\rbrace \right\rbrace
&\stackrel{\text{(a)}}=\min \left\lbrace \tDelta_{u^\star,A^c}~,~ \min_{u \in V^{'\star}_A } \left\lbrace \tDelta_{u,A^c} + S_{u,A}\right\rbrace \right\rbrace\\
&\stackrel{\text{(b)}}=\min \left\lbrace \tDelta_{u^\star,A^c}~,~ \tDelta_A \right\rbrace\label{eq:deltaAtoAu}\,,
\end{align}
where (a) follows from the fact that $\tDelta_{u^\star,A^c}\leq \tDelta_{u^\star,A^c}+S_{u^{\ast},A}$, and (b) follows from Equation~\eqref{eq:lem-age-by-path-simple}. Now we have two possibilities:
\begin{itemize}
    \item The optimal path in \eqref{eq:thm-main-delta-v-2} for $\tDelta_{u^\star}$ does not pass through $A$. In this case, we have $\tDelta_{u^\star,A^c}=\tDelta_{u^\star}$ and hence
    \begin{equation}
        \min \left\lbrace \tDelta_{u^\star,A^c}~,~ \tDelta_A \right\rbrace=\min \left\lbrace \tDelta_{u^\star}~,~ \tDelta_A \right\rbrace=\tDelta_{A\cup \{u^\star\}}\,.
    \end{equation}
    \item The optimal path in \eqref{eq:thm-main-delta-v-2} for $\tDelta_{u^\star}$ passes through $A$. In this case, we have $\tDelta_{u^\star,A^c}\geq \tDelta_{u^\star}\geq \tDelta_A$ and hence
    \begin{equation}
        \min \left\lbrace \tDelta_{u^\star,A^c}~,~ \tDelta_A \right\rbrace=\tDelta_A=\min \left\lbrace \tDelta_{u^\star}~,~ \tDelta_A \right\rbrace=\tDelta_{A\cup \{u^\star\}}\,.
    \end{equation}
\end{itemize}
We conclude that in all cases, we have
\begin{align}
\min \left\lbrace \tDelta_{u^\star,A^c}~,~ \min_{ \substack{u \in V^{'\star}_A: \\ u \neq u^\star} } \left\lbrace \tDelta_{u,A^c} + S_{u,A}\right\rbrace \right\rbrace
=\tDelta_{A\cup \{u^\star\}}\,.
\end{align}

By combining this with Equations~\eqref{eq:conditioning},~\eqref{eq:almost-there}, and~\eqref{eq:deltaAtoAu} yields our desired result in~\eqref{eq:EquationToShow-2}:
\begin{align}
\tilde{F}_A(s) & = \sum_{(u,v) \in E_A} \frac{\mu_{uv}}{\mu_A} \cdot \E \left[ \exp \left\lbrace s \tDelta_A \right\rbrace  \middle| S_A =S_{uv} \right] \\
& = \sum_{(u,v) \in E_A}  \frac{\mu_{uv}}{\mu_A} \cdot \frac{\mu_A}{\mu_A -s} \cdot\E \left[ \exp \left\lbrace s \tDelta_{A \cup \{u\}} \right\rbrace \right] \\
& = \sum_{(u,v) \in E_A}  \frac{\mu_{uv}}{\mu_A -s}\cdot \tilde{F}_{A \cup \{u\}}(s).
\end{align}

\section{Examples}~\label{sec:examples}

In this section, we apply our results in two basic, yet fundamental
layout examples. The first is the (well-studied and well-known) serial
cascade of servers and the second is the simple
``triangle'' layout of servers.

\subsection{Serial Cascade}

Consider a notwork consisting of a source $\src=v_0$, a destination $d=v_{n+1}$, and a serial cascade of $n$ relay servers $v_1,\ldots,v_n$. The cascade is shown in Figure~\ref{fig:Serial} below.

\begin{figure}[H]
  \begin{center}
    \setlength{\unitlength}{1.7cm}
    \begin{picture}(9.5,1)(0,-0.5)
      \put(0,0){\vtx{$\lambda$}{$v_0$}{$\mu_{v_0  v_1}$}}
      \put(2,0){\vtx{}{$v_1$}{$\mu_{v_1  v_2}$}}
      \dashline{0.07}(4.5,0)(6,0)
      \put(6,0){\vtx{$\mu_{v_{n-1}  v_n}$}{$v_n$}{}}
      \put(8,0){\vtxe{$\mu_{v_n  v_{n+1}}$}{$v_{n+1}$}}
    \end{picture}
    \caption{Serial cascade of servers.
      \label{fig:Serial} }
  \end{center}
\end{figure}
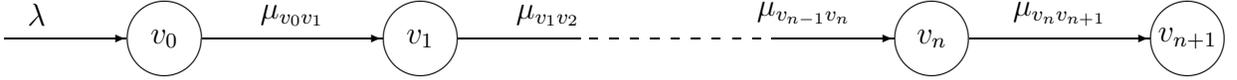

Yates used stochastic hybrid systems and showed
in~\cite{YatesPreempServs2018} that the average age at the destination is
\begin{equation}
  \frac{1}{\lambda} + \sum_{j = 0}^{n} \frac{1}{\mu_{v_{j}  v_{j+1}}} \,.
\end{equation}

Applying Theorem~\ref{thm:main} iteratively on the vertices $ v_1,
\cdots, v_n$ and $d=v_{n+1}$ yields
\begin{eqnarray}
  \Delta_{\src}& = &\Delta_{v_0}\sim \tDelta_{v_0} \sim  \mathrm{Exp}(\lambda) \\
  \Delta_{v_1} & \sim & \tDelta_{v_1}= \tDelta_{v_0} + S_{v_0 v_1} \sim \mathrm{Exp}(\lambda) + \mathrm{Exp}(\mu_{v_0  v_1}) \\
  \Delta_{v_j} & \sim & \tDelta_{v_j} = \tDelta_{v_{j-1}} + S_{v_{j-1}  v_{j}} \sim \mathrm{Exp}(\lambda) 
  + \sum_{l = 0}^{j-1} \mathrm{Exp}(\mu_{v_{l}  v_{l+1}}), \quad j = 2, \cdots, n+1 \\
  \Delta_{d} & = & \Delta_{v_{n+1}}\sim \mathrm{Exp}(\lambda) 
  + \sum_{l = 0}^{n} \mathrm{Exp}(\mu_{v_{l}  v_{l+1}})\\
  \E[\Delta_{d}] & = & \frac{1}{\lambda} + \sum_{l = 0}^{n} \frac{1}{\mu_{v_{l}  v_{l+1}}} \,,
\end{eqnarray}
which are the results derived in~\cite{YatesPreempServs2018} and \cite{YatesAoIMomTrans2018}.

\subsection{Triangle}

\label{sec:Triangle}

Next, we consider the special case where the source and destination
form with an additional server a ``triangular'' layout as shown in
Figure~\ref{fig:Triangle} below.

\begin{figure}[H]
  \begin{center}
    \setlength{\unitlength}{1.5cm}
    \begin{picture}(4.25,2)(0,-0.4)
      \put(0,0){\vtx{$\lambda$}{$\src$}{$\mu_{\src  d}$}}
      \put(1.25,0.3){\vtxud{$\mu_{\src  v}$}{$v$}{$\mu_{v  d}$}}
      \put(2.5,0){\vtxe{}{$d$}}
    \end{picture}
    \caption{Triangle layout.
      \label{fig:Triangle} }
  \end{center}
\end{figure}
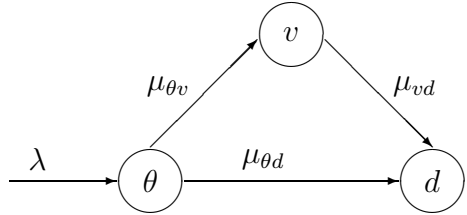

We apply Theorem~\ref{thm:main} on the vertices $\{ v, d \}$ to deduce
the distribution of the age at the destination:
\begin{eqnarray}
  \Delta_{\src} & \sim & \tDelta_{\src}\sim \mathrm{Exp}(\lambda) \\
  \Delta_{v} & \sim & \tDelta_{v} =  \tDelta_{\src} + S_{\src  v} \sim \mathrm{Exp}(\lambda) + \mathrm{Exp}(\mu_{\src  v}) \\
  \Delta_{d} & \sim & \tDelta_{d} =  \min \left\{ \tDelta_{\src} + S_{\src  d}, \tDelta_{v} + \
  S_{v  d} \right\} \\
  & \sim & \mathrm{Exp}(\lambda) + \min \left\{ \mathrm{Exp}(\mu_{\src  d})~,~ \mathrm{Exp}(\mu_{\src  v}) + \mathrm{Exp}(\mu_{v  d}) \right\}.
\end{eqnarray}

Next we determine the distribution of the minimum. Let $X$ denote the
random variable that is the minimum of $S_{\src  d}\sim\mathrm{Exp}(\mu_{\src  d})$
and $S \sim \mathrm{Exp}(\mu_{\src  v}) + \mathrm{Exp}(\mu_{v 
  d})$.

Assume first that $\mu_{\src  v}$ and $\mu_{v  d}$ are
different. Noting that the random variable $S$ is hypoexponentially\footnote{The cdf of a hypoexponential random variable with (different) parameters $\mu_1$ and $\mu_2$ is given by $F(t) =1- \frac{\mu_2 }{\mu_2 - \mu_1} \exp \{- \mu_1 t\}+ \frac{\mu_1 }{\mu_2 - \mu_1} \exp \{- \mu_2 t\}$, $t \geq 0$.}
distributed with parameters $\{ \mu_{\src  v}, \mu_{v  d}\}$ and
using independence, we get
\begin{align}
  \Pr(X > x) & = \Pr [S_{\src  d} > x]  \cdot \Pr [S > x] \\
  & = e^{- \mu_{\src  d} \, x}\cdot \left( \frac{\mu_{v  d}}{\mu_{v  d} - \mu_{\src  v}} e^{-\mu_{\src  v} \, x}
    - \frac{\mu_{\src  v}}{\mu_{v  d} - \mu_{\src  v}} e^{- \mu_{v  d} \, x} \right) \\
  & =  \frac{\mu_{v  d}}{\mu_{v  d} - \mu_{\src  v}} e^{- (\mu_{\src  v} + \mu_{\src  d} ) \, x}
  - \frac{\mu_{\src  v}}{\mu_{v  d} - \mu_{\src  v}} e^{- (\mu_{v  d} + \mu_{\src  d} ) \, x}\,, 
 \end{align}
 hence, the pdf of $X$ satisfies
 \begin{align}
  p_X(x) & = \frac{\mu_{v  d}}{\mu_{v  d} - \mu_{\src  v}} p_{Y}(x) 
  - \frac{\mu_{\src  v}}{\mu_{v  d} - \mu_{\src  v}} p_{Z}(x)\,,
\end{align}
where $Y\sim\mathrm{Exp}(\mu_{\src  v} + \mu_{\src  d} )$ and $Z\sim\mathrm{Exp}(\mu_{v  d} + \mu_{\src  d} )$. We can now deduce that the average AoI at destination is
\begin{align}
  \E[\Delta_{d}] & = \E[\tDelta_\src] + \frac{\mu_{v  d}}{\mu_{v  d} - \mu_{\src  v}}\,\E[Y] - \frac{\mu_{\src  v}}{\mu_{v  d} - \mu_{\src  v}}\,\E[Z]\\
  & = \frac{1}{\lambda} + \frac{\mu_{v  d}}{\mu_{v  d} - \mu_{\src  v}} \, 
  \frac{1}{\mu_{\src  v} + \mu_{\src  d}}
  - \frac{\mu_{\src  v}}{\mu_{v  d} - \mu_{\src  v}}  \frac{1}{\mu_{v  d} + \mu_{\src  d}} \\
  & = \frac{1}{\lambda} + \frac{\mu_{v  d} \left( \mu_{v  d} + \mu_{\src  d} \right)
    - \mu_{\src  v} \left( \mu_{\src  v} + \mu_{\src  d} \right)}{\left( \mu_{v  d} - \mu_{\src  v} \right)
    \left( \mu_{\src  v} + \mu_{\src  d} \right) \left( \mu_{v  d} + \mu_{\src  d} \right)} \\
  & = \frac{1}{\lambda} + \frac{\mu_{\src  v} + \mu_{v  d} + \mu_{\src  d} }{
    \left( \mu_{\src  v} + \mu_{\src  d} \right) \left( \mu_{v  d} + \mu_{\src  d} \right)} \\
  & = \frac{1}{\lambda} + \frac{1}{\mu_{\src  v} + \mu_{\src  d}}
  + \frac{1}{\mu_{v  d} + \mu_{\src  d} }
  - \frac{\mu_{\src  d} }{\left( \mu_{\src  v} + \mu_{\src  d} \right) \left( \mu_{v  d} + \mu_{\src  d} \right)}\,.
\end{align}

In the case where $\mu_{\src  v} = \mu_{v  d} = \mu$, the random
variable $S$ is Erlang-2 distributed with (rate) parameter $\mu$. Using
independence,
\begin{align}
  \Pr[X > x] & = \Pr [S_{\src  d} > x]  \cdot \Pr [S > x] \\
  & = e^{- \mu_{\src  d} \, x} \left( 1 + \mu x  \right) e^{- \mu x} \\
  & = \left( 1 + \mu x  \right) e^{- (\mu + \mu_{\src  d}) x}\,,
   \end{align}
 and the pdf of $X$ satisfies
 \begin{align}
  p_X(x) & = \left( \frac{\mu_{\src  d}}{\mu + \mu_{\src  d}} + \mu x \right) p_{Y}(x)\,,
\end{align}
where $Y\sim \mathrm{Exp}(\mu + \mu_{\src  d} )$. We can now get the ``same'' above formula for the average AoI at the destination:
\begin{align}
  \E[\Delta_{d}] 
  & = \E[\tDelta_\src] + \frac{\mu_{\src  d}}{\mu + \mu_{\src  d}}\,\E[Y] + \mu \,\E[Y^2]\\
  & = \frac{1}{\lambda} + \frac{\mu_{\src  d}}{\left( \mu + \mu_{\src  d} \right)
    \left(\mu + \mu_{\src  d} \right) } + \frac{2 \mu}{\left(\mu + \mu_{\src  d} \right)^2 } \\
  & = \frac{1}{\lambda} + \frac{\mu_{\src  d} + 2 \mu}{\left( \mu + \mu_{\src  d} \right)
    \left(\mu + \mu_{\src  d} \right) } \\
  & = \frac{1}{\lambda} + \frac{2}{\left( \mu + \mu_{\src  d} \right) }
  - \frac{\mu_{\src  d}}{\left(\mu + \mu_{\src  d} \right)^2 }\,.
\end{align}

Note that the above formulas for the average AoI can also be obtained using the recursive formula for the average AoI that Yates derived in \cite{YatesGossipIWSPAWC2021}.

\section{Applications and computational aspects}~\label{sec:computation}

In this section we discuss applications and a few computational aspects of our results.

\subsection{MGF of the stationary distribution of AoI}

\label{sec:computingMGF}

By replacing $F_{A'}$ in Equation~\eqref{eq:lem-MGF-marg-recur} with its definition
\begin{equation}
F_{A'}\left(s \right) =\E \left[ \exp\left(s \Delta_{A'}\right) \right]\,,
\end{equation}
we get the following recursive formula for every $\src\notin A\in\V$:
\begin{equation}
\label{eq:MGF-Relations-recur}
\E \left[ \exp\left(s \Delta_{A}\right) \right]  =  \sum_{(u,v) \in E_{A} } \frac{\mu_{uv}}{\mu_{A} - s}\cdot \E \left[ \exp\left(s \Delta_{A\cup\{u\}}\right) \right]\,,
\end{equation}
where $E_{A}$ and $\mu_{A}$ are as in Equation~\eqref{eq:def-E_A} and \eqref{eq:def-mu_A}, respectively.

The recursive formula in~\eqref{eq:MGF-Relations-recur} can be used to compute the MGFs of all AoIs: We can use the fact that Equations \eqref{eq:lem-MGF-src-prime} and \eqref{eq:lem-MGF-src} give us $\E \left[ \exp\left(s \Delta_{A}\right) \right] =F_A(s)$ for $\src\in A$ or $\src\notin A$, and starting from this, we can recursively compute $\E \left[ \exp\left(s \Delta_{A}\right) \right]$ by order of decreasing $|A|$. The computational complexity is exponential in the number of nodes in the graph, which is essentially the same complexity as the method obtained by Yates in \cite{YatesGossipIWSPAWC2021} for computing average AoIs.

\subsection{Exact computation of average AoI}
In this subsection, we re-derive the recursive formula of \cite{YatesGossipIWSPAWC2021} for average AoIs. By taking the derivative with respect to $s$ on both sides of Equation~\eqref{eq:MGF-Relations-recur}, and then evaluating at $s=0$, we get:
\begin{align}
\label{eq:AoI-Relations-recur}
\E \left[ \Delta_{A}\right] 
&=  \sum_{(u,v) \in E_{A} } \left( \frac{\mu_{uv}}{\mu_{A} }\cdot \E \left[ \Delta_{A\cup\{u\}} \right]+  \frac{\mu_{uv}}{\mu_{A}^2 }\right)\\
&=\frac{\mu_{A}}{\mu_{A}^2 }+  \sum_{(u,v) \in E_{A} } \left( \frac{\mu_{uv}}{\mu_{A} }\cdot \E \left[ \Delta_{A\cup\{u\}} \right]\right)\\
&=  \frac{1}{\mu_A}\left( 1+\sum_{(u,v) \in E_{A} } \mu_{uv}\cdot \E \left[ \Delta_{A\cup\{u\}} \right]\right)\,.
\end{align}
By noticing that $\Delta_{A\cup\{\src'\}}=0$, we can see that the above formula is exactly the same as the one obtained by Yates in \cite{YatesGossipIWSPAWC2021}.

\subsection{Exact computation of averages of arbitrary functions of the age}

Once the MGF of the AoI is computed using~\eqref{eq:MGF-Relations-recur}, we can compute the probability distribution by applying inverse Fourier transform. This allows for the computation of the average of an arbitrary function of the age. In particular, we can compute the age-violation probability
\begin{equation}
\Pr[\Delta_{A}\geq d] =  \E \left[ \mathbf{1}_{\{\Delta_{A}\geq d\}}\right] \,.
\end{equation}

If we are only interested in getting upper bounds on the age-violation probability, we can avoid the application of inverse Fourier transform. More precisely, we can directly use the moment-generating function and apply the Chernoff bound.

\subsection{Simplifications for exact computation in structured networks}

If the SSN $G$ is structured, then the simple characterizations in Theorems~\ref{thm:main} and~\ref{thm:statDistSubsets} might allow us to leverage the structure of $G$ in order to simplify the derivation of AoI distribution.

For example, consider the following cascade of two triangles $(v_0,v_1,v_2)$ and $(v_2,v_3,v_4)$, connecting a source $\src=v_0$ and a destination $d=v_4$:

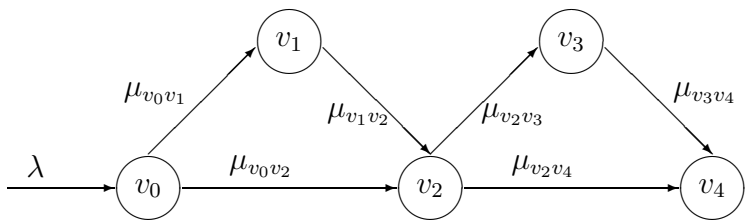
\begin{figure}[H]
  \begin{center}
    \setlength{\unitlength}{1.5cm}
    \begin{picture}(6.75,2)(0,-0.4)
      \put(0,0){\vtx{$\lambda$}{$v_0$}{$\mu_{v_0  v_2}$}}
      \put(1.25,0.3){\vtxudq{$\mu_{v_0 v_1}$}{$v_1$}{$\mu_{v_1  v_2}$}}
      \put(2.5,0){\vtx{}{$v_2$}{$\mu_{v_2  v_4}$}}
      \put(3.75,0.3){\vtxudp{$\mu_{v_2 v_3}$}{$v_3$}{$\mu_{v_3  v_4}$}}
      \put(5,0){\vtxe{}{$v_4$}}
    \end{picture}
    \caption{Two triangles layout.
      \label{fig:TwoTriangles} }
  \end{center}
\end{figure}

Using the formulas of Theorems~\ref{thm:main} and~\ref{thm:statDistSubsets}, and the computations of 
Section~\ref{sec:Triangle}, we can easily get:
\begin{align}
  \E[\Delta_{v_2}] & = \E[\tDelta_{v_2}] =\frac{1}{\lambda} + \frac{\mu_{v_0  v_1} + \mu_{v_1  v_2} + \mu_{v_0  v_2} }{
    \left( \mu_{v_0  v_1} + \mu_{v_0  v_2} \right) \left( \mu_{v_1  v_2} + \mu_{v_0  v_2} \right)}\,, \\
\end{align}
and
\begin{align}
  \E[\Delta_{v_4}] & = \E[\tDelta_{v_4}] = \E[\tDelta_{v_2}]+\E[\min\{S_{v_2v_4},S_{v_2v_3}+S_{v_3v_4}\}] \\
  & = \frac{1}{\lambda} + \frac{\mu_{v_0  v_1} + \mu_{v_1  v_2} + \mu_{v_0  v_2} }{
    \left( \mu_{v_0  v_1} + \mu_{v_0  v_2} \right) \left( \mu_{v_1  v_2} + \mu_{v_0  v_2} \right)}+ \frac{\mu_{v_2  v_3} + \mu_{v_3  v_4} + \mu_{v_2  v_4} }{
    \left( \mu_{v_2  v_3} + \mu_{v_2  v_4} \right) \left( \mu_{v_3  v_4} + \mu_{v_2  v_4} \right)}\,. \\
\end{align}

We basically used the fact that every path from $v_0$ to $v_4$ must pass by $v_2$.

Similarly, if we have a cascade of $n$ triangles $(v_{2i-2},v_{2i-1},v_{2i})_{1\leq i\leq n}$ connected a source $\src = v_0$ and $d=v_{2n}$, we can use the formulas of Theorems~\ref{thm:main} and~\ref{thm:statDistSubsets}, and the computations of 
Section~\ref{sec:Triangle} to get:
\begin{align}
  \E[\Delta_{v_{2n}}] = \E[\tDelta_{v_{2n}}] = \frac{1}{\lambda} + \sum_{i=1}^n \frac{\mu_{v_{2i-2}  v_{2i-1}} + \mu_{v_{2i-1}  v_{2i}} + \mu_{v_{2i-2}  v_{2i}} }{
    \left( \mu_{v_{2i-2}  v_{2i-1}} + \mu_{v_{2i-2}  v_{2i}} \right) \left( \mu_{v_{2i-1}  v_{2i}} + \mu_{v_{2i-2}  v_{2i}} \right)}\,.
\end{align}

More generally, assume that the SSN $G=(V,E)$ satisfies the following properties:
\begin{enumerate}
    \item There exists a size-$n$ cover $\{V_1,\ldots,V_n\}$ of the set of nodes $V$ of the network $G$, i.e., 
\begin{equation}
V=\bigcup_{i=1}^{n} V_i\,.
\end{equation}
\item $\src\in V_1$.
\item $V_i\cap V_{j}\neq\varnothing$ if and only if $j\in\{i,i+1\}$.
\item There exists a universal constant $C$, which does not scale with $n$, such that for every $1\leq i\leq n$, we have
\begin{equation}
|V_i|\leq C\,.
\end{equation}
\item For every edge $(u,v)\in E$, there exists $1\leq i\leq n$ such that $u\in V_i$ and $v\in V_i\cup V_{i+1}$.
\end{enumerate}
For such a network, any path from $\src$ to a vertex $v\in V_i$ must subsequently pass through vertices in $V_1\cap V_2$, $V_2\cap V_3$, \ldots, and $V_{i-1}\cap V_i$. Furthermore, once a path passes from $V_{j-1}\setminus V_j$ to $V_j$, it cannot go back. Based on these observations, it is not hard to see that we can compute the AoI distributions in $n$ stages as follows:
\begin{itemize}
    \item In stage 1, we compute the joint distribution of $(\tDelta_v)_{v\in V_1}$ using the fact that
    \begin{equation}
        \tDelta_v = \min_{P \in \mathcal{P}(\src' \rightarrow v)} \sum_{e \in P} S_e\,,\quad\text{for all }v\in V_1\,.
    \end{equation}
    \item By stage $2\leq i\leq n$, we will have obtained the joint distribution of $(\tDelta_v)_{v\in V_{i-1}}$. So we marginalize and obtain the joint distribution of $(\tDelta_v)_{v\in V_{i-1}\cap V_i}\,$, from which we compute the joint distribution of $(\tDelta_v)_{v\in V_i}$ using the fact that
    \begin{equation}
        \tDelta_v = \min_{u\in V_{i-1}\cap V_i}\left\lbrace\tDelta_u + \min_{P \in \mathcal{P}(u \rightarrow v)} \sum_{e \in P} S_e\right\rbrace \,,\quad\text{for all }v\in V_i\,.
    \end{equation}
\end{itemize}
Using marginalization, we obtain the distribution of $\tDelta_v$, which is the same as that of $\Delta_v$, for every node $v\in V$. Taking expectations, we get the average AoIs.

This procedure will take $n$ elementary computations, which might be better than the naive application of the exponential time procedure of Section~\ref{sec:computingMGF} for general unstructured networks. Note that by an elementary computation, we mean the completion of one stage. Now since $|V_i|$ can be as large as $C$, we can see that the time complexity of one elementary computation can be exponential in $C$. But $C$ is an absolute constant that does not scale with $n$. Hence, if we are interested in the asymptotic growth of the computation time in terms of $n$, we can ignore the dependence on $C$.

The above discussion might give the false impression that the described procedure has an overall time complexity that is linear in $n$. However, this is not necessarily the case because the computation time of the $i$-th stage depends on the length of the representation of the joint distribution of $(\tDelta_v)_{v\in V_{i-1}}$, and this can grow exponentially in $i$. Nevertheless, the proper implementation of the above procedure might be helpful in reducing the computation time in some cases. For example, if $|V_i\cap V_{i+1}|=1$ for all $1\leq i<n$, we can definitely compute the average AoIs $\big(\E[\Delta_v]\big)_{v\in V}$ for all nodes in linear time.

\subsection{Faster and more accurate Monte Carlo simulations}

If we are only interested in estimating the average AoIs through Monte Carlo simulations, Theorem~\ref{thm:main} provides a very simple and effective procedure:
\begin{itemize}
    \item Generate $N$ i.i.d. tuples $\left(S_e^{(i)}\right)_{e\in E'}\,, 1\leq i\leq N$, where $\left(S_e^{(i)}\right)_{e\in E'}$ are mutually independent and $S_e^{(i)}\sim\mathrm{Exp}(\mu_e)$. 
    \item For every $1\leq i\leq N$, use Dijkstra's algorithm to compute
    \begin{equation}
        \tilde{\Delta}_v^{(i)} = \min_{P \in \mathcal{P}(\src' \rightarrow v)} \sum_{e \in P} S_e^{(i)}, \text{ for every } v \in V\,.
    \end{equation}
    This requires $O\big((|E|+|V|)\log|V|\big)$ computations.
    \item For every $v\in V$, we estimate $\Delta_v$ as
    \begin{equation}
        \hDelta_v = \frac{1}{N}\sum_{i=1}^N\tilde{\Delta}_v^{(i)}\,.
    \end{equation}
\end{itemize}
The overall procedure requires $O\big(N(|E|+|V|)\log|V|\big)$ computations. Note that if we are interested in estimating the average AoI for a single node $v\in V$, we do not need to run Dijkstra's algorithm for the entire network: We keep running until hitting $v$, and we generate $S_e^{(i)}\sim\mathrm{Exp}(\mu_e)$ as we go. This will save us some unnecessary computational steps.

Now of course, we could simulate the actual AoI mechanism without the need of Theorem~\ref{thm:main}: We generate $N_p$ events of a $\mathrm{Poisson}(\lambda)$ process, corresponding to $N_p$ packets at the source. Then, we propagate the packets through the network while applying the preemptive mechanism that is described at the beginning of Section~\ref{sec:MainResult}. However, for a target accuracy $\epsilon$ and a target confidence probability $1-\delta$, we expect that such a simulation would require more computations (compared to the one based on Theorem~\ref{thm:main}) since we would need to wait until the age processes converge to their stationary distributions, and since AoIs at close instants of time are very correlated.

\section{Conclusion}

This paper derives a very simple characterization of the stationary distribution of AoI at every node in a network with memoryless service times, where all nodes follow a preemptive policy. The insights that we obtained from this characterization can substantially reduce the computation time for the average AoI. This is true both for exact computations and for Monte Carlo simulations.

One important extension of this work would be to consider multiple-source memoryless networks where status updates are generated at various source nodes, and then forwarded across the network. If each link can transmit only one or a limited number of packets at a time, we need to specify a policy that determines which packet to transmit in case a node has packets from several sources. For multiple-source memoryless networks, there is no single policy that is optimal for all AoIs\footnote{This is in contrast with single-source memoryless networks for which there is an age-optimal policy, namely preemption in service.}: We have tradeoffs between the AoIs corresponding to different sources. It is important to determine the policies that achieve these optimal tradeoffs. Furthermore, it would be nice if one can obtain a recipe to compute (or at least bound) the average AoIs for these policies.

Another possible extension of this work would be to consider single-source networks for which the interarrival time and/or the service times are not memoryless. Some of the techniques that are developed in this paper might be useful, but we expect that further tools are needed in order to completely characterize the stationary distribution and/or the average of AoI for the preemptive policy in such general settings. It is worth noting here that when the service times are not memoryless, the preemptive policy might not be age-optimal. Therefore, it makes sense to consider other transmission policies for these general settings.

\section*{Acknowledgment}
We would like to thank Yunus Inan and Emre Telatar for helpful discussions.

\appendices

\section{Proof of Lemma~\ref{lem:MGF-Relations}}
\label{app:proof-lem-MGF-Relations}

We have
\begin{align}
F\left((s_A)_{A\in\V'}\right) &= \E \left[ \exp\left(\sum_{A\in \V' }s_A \Delta_A\right) \right] \\
& \stackrel{\text{(a)}} = \E \left[ \exp\left(\sum_{A\in \V }s_A \Delta_A\right) \right] \\
& \stackrel{\text{(b)}} = \lim_{\tau\to\infty}\frac{1}{\tau}\int_{0}^{\tau} \exp\left(\sum_{A\in\V}s_A \Delta_A(t)\right) dt\\
&=\lim_{n\to\infty}\frac{1}{B_n}\int_{0}^{B_n} \exp\left(\sum_{A\in\V}s_A \Delta_A(t)\right) dt\\
&=\lim_{n\to\infty}\frac{n}{B_n}\cdot\frac{1}{n}\sum_{i=1}^n\int_{B_{i-1}}^{B_i} \exp\left(\sum_{A\in\V}s_A \Delta_A(t)\right) dt\\
& \stackrel{\text{(c)}} =\frac{1}{\E[T]}\cdot\lim_{n\to\infty}\E\left[\int_{B_{n-1}}^{B_n} \exp\left(\sum_{A\in\V}s_A \Delta_A(t)\right) dt\right]\,, \label{eq:F-1}
\end{align}
where (a) follows from the fact that $\src' \in A$ implies $\Delta_A = 0$, (b) follows from ergodicity, and (c) follows from the (strong) law of large numbers. Note that
\begin{equation} \label{eq:ET}
    \E[T]=\frac{1}{\mu}\,,
\end{equation}
and for every $n\geq 1$, we have
\begin{align}
\E&\left[\int_{B_{n-1}}^{B_n} \exp\left(\sum_{A\in\V}s_A \Delta_A(t)\right) dt\right]\\
&=\E\left[\int_{B_{n-1}}^{B_n} \exp\left(\sum_{A\in\V}s_A (\Delta_A(B_{n-1})+t-B_{n-1})\right) dt\right]\\
&=\E  \left[\exp\left(\sum_{A\in\V}s_A \Delta_A(B_{n-1}) \right) \cdot \int_{0}^{T_n} \exp\left(\sum_{A\in\V}s_A t\right) dt\right]\\
&=\E\left[\frac{\exp\left(\resizebox{0.03\textwidth}{!}{$\displaystyle\sum_{A\in\V}$}~s_A \Delta_A(B_{n-1})\right)}{\resizebox{0.03\textwidth}{!}{$\displaystyle\sum_{A\in\V}$}~s_A} \cdot\left(\exp\left(\sum_{A\in\V}s_AT_{n}\right)-1 \right)\right]\\
&=\frac{\E\left[\exp\left(\resizebox{0.03\textwidth}{!}{$\displaystyle\sum_{A\in\V}$}~s_A \Delta_A(B_{n-1})\right)\right]}{\resizebox{0.03\textwidth}{!}{$\displaystyle\sum_{A\in\V}$}~s_A} \cdot\E\left[\exp\left(\sum_{A\in\V}s_AT_{n}\right)-1 \right]\,, 
\end{align}
where the last equality follows from the fact that $T_n$ is independent from $(\Delta_A(B_{n-1}))_{A\in\V}$. Therefore,
\begin{align}
\E\left[\int_{B_{n-1}}^{B_n} \exp\left(\sum_{A\in\V}s_A \Delta_A(t)\right) dt\right]
&=\frac{\E\left[\exp\left(\resizebox{0.03\textwidth}{!}{$\displaystyle\sum_{A\in\V}$}~s_A \Delta_A(B_{n-1})\right)\right]}{\resizebox{0.03\textwidth}{!}{$\displaystyle\sum_{A\in\V}$}~s_A} \cdot\left(\frac{\mu}{\mu - \resizebox{0.03\textwidth}{!}{$\displaystyle\sum_{A\in\V}$}~s_A} - 1\right)\\
&=\frac{\E\left[\exp\left(\resizebox{0.03\textwidth}{!}{$\displaystyle\sum_{A\in\V}$}~s_A \Delta_A(B_{n-1})\right)\right]}{\mu - \resizebox{0.03\textwidth}{!}{$\displaystyle\sum_{A\in\V}$}~s_A}\,. \label{eq:F-2}
\end{align}
Now from~\eqref{eq:F-1},~\eqref{eq:ET}, and~\eqref{eq:F-2}, we conclude that
\begin{equation}
\label{eq:relationFtof}
F\left((s_A)_{A\in\V'}\right) =\frac{\mu}{\mu - \resizebox{0.03\textwidth}{!}{$\displaystyle\sum_{A\in\V}$}~s_A}\cdot f\left((s_A)_{A\in\V}\right)\,.
\end{equation}

\section{Proof of Lemma~\ref{lem:MGF-Relations-recur}}
\label{app:proof-lem-MGF-Relations-recur}


We first derive a recursive relation that is satisfied by the function $f$ of \eqref{eq:MGF-critical}. Note that for every $n\geq 1$, we have
\begin{align}
&\E\left[\exp\left(\sum_{A\in\V}s_A \Delta_A(B_{n})\right)\right]\\
& = \sum_{e\in E'}\Pr[E_n=e]\cdot \E\left[\exp\left(\sum_{A\in\V}s_A \Delta_A(B_{n})\right)\middle|E_n=e\right] \\
& \stackrel{\text{(a)}}=\frac{\mu_{\src'\src}}{\mu}\cdot \E\left[\exp\left(\sum_{\substack{A\in\V:\\ \src\notin A}}s_A (\Delta_{A}(B_{n-1})+T_{n})\right)\right]\\
&\quad\quad+\sum_{(u,v)\in E}\frac{\mu_{uv}}{\mu}\cdot \E\left[\exp\left(\sum_{\substack{A\in\V:\\ v\notin A}}s_A (\Delta_{A}(B_{n-1})+T_{n}) + \sum_{\substack{A\in\V:\\ v\in A}}s_A (\Delta_{A\cup \{u\}}(B_{n-1})+T_{n})\right)\right]\\
& \stackrel{\text{(b)}}=\resizebox{0.53\textwidth}{!}{$\displaystyle\frac{\mu_{\src'\src}}{\mu}\cdot \E\left[\exp\left(\sum_{\substack{A\in\V:\\ \src\notin A}}s_A \Delta_{A}(B_{n-1})\right)\right]\cdot  \E\left[\exp\left(\sum_{\substack{A\in\V:\\ \src\notin A}}s_AT_{n}\right)\right]$}\\
&\quad+\resizebox{0.83\textwidth}{!}{$\displaystyle\sum_{(u,v)\in E'}\frac{\mu_{uv}}{\mu}\cdot \E\left[\exp\left(\sum_{\substack{A\in\V:\\ v\notin A}}s_A \Delta_{A}(B_{n-1}) + \sum_{\substack{A\in\V:\\ v\in A}}s_A \Delta_{A\cup \{u\}}(B_{n-1})\right)\right]\cdot \E\left[\exp\left(\sum_{\substack{A\in\V}}s_AT_{n}\right)\right]$}\,,
\end{align}
where (a) follows from the update equation~\eqref{eq:updateGivenEnEqualeA} and the fact that $\big((\Delta_A(B_{n-1}))_{A\in \V},T_n\big)$ is independent
from $E_n$, and (b) follows from the fact that $T_n$ is independent from $(\Delta_A(B_{n-1}))_{A\in\V}$. Therefore,
\begin{align}
&\E\left[\exp\left(\sum_{A\in\V}s_A \Delta_A(B_{n})\right)\right]\\
&=\frac{\mu_{\src'\src}}{\mu}\cdot \E\left[\exp\left(\sum_{\substack{A\in\V:\\ \src\notin A}}s_A \Delta_{A}(B_{n-1}) \right)\right]\cdot \frac{\mu}{\mu - \resizebox{0.036\textwidth}{!}{$\displaystyle\sum_{\substack{A\in\V:\\ \src\notin A}}$}~s_A}\\
&\quad+\sum_{(u,v)\in E}\frac{\mu_{uv}}{\mu}\cdot \E\left[\exp\left(\sum_{\substack{A\in\V:\\ v\notin A}}s_A \Delta_{A}(B_{n-1}) + \sum_{\substack{A\in\V:\\ v\in A}}s_A \Delta_{A\cup \{u\}}(B_{n-1})\right)\right]\cdot \frac{\mu}{\mu - \resizebox{0.03\textwidth}{!}{$\displaystyle\sum_{A\in\V}$}~s_A}\,.
\end{align}
By taking the limit as $n\to\infty$, we get
\begin{equation}
\begin{aligned}
  f\left((s_A)_{A\in\V}\right) &=\frac{\mu_{\src'\src}}{\mu - \resizebox{0.036\textwidth}{!}{$\displaystyle\sum_{\substack{A\in\V:\\ \src\notin A}}$}~s_A}\cdot  f\left((s_A)_{\substack{A\in\V:\\\src\notin A}}~,~(0)_{\substack{A\in\V:\\\src\in A}}\right)\\
 &\quad \quad +\sum_{(u,v)\in E}\frac{\mu_{uv}}{\mu - \resizebox{0.03\textwidth}{!}{$\displaystyle\sum_{A\in\V}$}~s_A}\cdot f\left((s_A)_{\substack{A\in\V:\\v\notin A}}~,~(0)_{\substack{A\in\V:\\v\in A,\\u\notin A }}~,~(s_A+s_{A\setminus\{u\}})_{\substack{A\in\V:\\u,v\in A}}\right)\,.
\end{aligned}
  \label{eq:lem-MGF-critical-recur}
  \end{equation}



Now consider an arbitrary $A' \in \V'$. We have:
\begin{itemize}
\item If $\src'\in A'$, then $\Delta_{A'}=0$ and
\begin{align}
F_{A'}\left(s \right) & = F\left((s)_{A=A'}~,~(0)_{\substack{A\in\V':\\A\neq A'}}\right) = \E\left[\exp(s\Delta_{A'})\right]=1\,.
\end{align}
\item If $\src\in A'\in\V$, we would like to show that $\ds F_{A'}\left(s \right)=\frac{\mu_{\src'\src}}{\mu_{\src'\src} -s}$. By defining
\begin{align}
f_{A'}\left(s \right) & := f\left((s)_{A=A'}~,~(0)_{\substack{A\in\V:\\A\neq A'}}\right)\,,
\end{align} 
we can see that in the light of Lemma~\ref{lem:MGF-Relations}, it is sufficient to show that
\begin{align}
f_{A'}\left(s \right) & :=\frac{\mu_{\src'\src}(\mu-s)}{\mu(\mu_{\src'\src} -s)}\,.
\end{align}
We will show this for every $A'$ satisfying $\src\in A'\in\V$ by reverse induction on $|A'|$:
\begin{itemize}
\item If $|A'| = |V|$, i.e., if $A'=V$, then by setting $s_{V}=s$ and $s_A=0$ for all $A\neq V$ in \eqref{eq:lem-MGF-critical-recur}, we get
\begin{align}
f_{V}\left(s \right) &= \frac{\mu_{\src'\src}}{\mu}+\sum_{(u,v)\in E} \frac{\mu_{uv}}{\mu - s} f_{V}\left(s \right),
\end{align} 
where we used the fact that $f\left((0)_{A\in\V}\right)=1$. By rearranging the above equation, we get
\begin{align}
f_{V}\left(s \right) &= \frac{\mu_{\src'\src}(\mu-s)}{\mu(\mu_{\src'\src} -s)}\,.
\end{align} 
\item Let $\src\in A'\in\V$ be such that $|A'| < |V|$, and assume that
\begin{align}
f_{A}\left(s \right) &= \frac{\mu_{\src'\src}(\mu-s)}{\mu(\mu_{\src'\src} -s)}\,,\text{ for all }A\in\V\text{ satisfying }\src\in A\text{ and }|A|>|A'|\,.
\label{eq:ref-induction-f}
\end{align} 
By setting $s_{A'}=s$ and $s_A=0$ for all $A\neq A'$ in \eqref{eq:lem-MGF-critical-recur}, we get
\begin{align}
f_{A'}\left(s \right) &= \frac{\mu_{\src'\src}}{\mu}+\sum_{\substack{(u,v)\in E:\\v\notin A'\text{ or }u\in A'}} \frac{\mu_{uv}}{\mu - s} f_{A'}\left(s \right) + \sum_{\substack{(u,v)\in E:\\v\in A'\text{ and }u\notin A'}} \frac{\mu_{uv}}{\mu - s} f_{A'\cup\{u\}}\left(s \right)\\
&= \frac{\mu_{\src'\src}}{\mu}+\sum_{\substack{(u,v)\in E:\\v\notin A'\text{ or }u\in A'}} \frac{\mu_{uv}}{\mu - s} f_{A'}\left(s \right) + \sum_{\substack{(u,v)\in E:\\v\in A'\text{ and }u\notin A'}} \frac{\mu_{uv}}{\mu - s} \cdot \frac{\mu_{\src'\src}(\mu-s)}{\mu(\mu_{\src'\src} -s)}\,,
\end{align} 
where the first equality follows from the fact that $f\left((0)_{A\in\V}\right)=1$, and the second equation follows from the induction hypothesis \eqref{eq:ref-induction-f}. By rearranging the above equation, we get
\begin{align}
f_{A'}\left(s \right) &= \frac{\mu_{\src'\src}(\mu-s)}{\mu(\mu_{\src'\src} -s)}\,.
\end{align} 
\end{itemize}
We conclude that for every $\src\in A'\in\V$, we have
\begin{align}
f_{A'}\left(s \right) = \frac{\mu_{\src'\src}(\mu-s)}{\mu(\mu_{\src'\src} -s)}\quad\text{and}\quad F_{A'}\left(s \right) = \frac{\mu_{\src'\src}}{\mu_{\src'\src} -s}\,.
\end{align} 
\item If $\src\notin A'\in\V$, then by setting $s_{A'}=s$ and $s_A=0$ for all $A\neq A'$ in~\eqref{eq:lem-MGF-critical-recur}, we get
\begin{align}
f_{A'}\left(s \right) & = f\left((s)_{A=A'}~,~(0)_{\substack{A\in\V:\\A\neq A'}}\right) \\
&= \frac{\mu_{\src'\src}}{\mu - s}  f_{A'}\left(s \right) +\sum_{\substack{(u,v)\in E:\\v\notin A'\text{ or }u\in A'}} \frac{\mu_{uv}}{\mu - s} f_{A'}\left(s \right) + \sum_{\substack{(u,v)\in E:\\v\in A'\text{ and }u\notin A'}} \frac{\mu_{uv}}{\mu - s} f_{A'\cup\{u\}}\left(s \right)\\
&= \sum_{\substack{(u,v)\in E':\\v\notin A'\text{ or }u\in A'}} \frac{\mu_{uv}}{\mu - s} f_{A'}\left(s \right) + \sum_{\substack{(u,v)\in E':\\v\in A'\text{ and }u\notin A'}} \frac{\mu_{uv}}{\mu - s} f_{A'\cup\{u\}}\left(s \right)\,.
\end{align} 
Combining this with Lemma~\ref{lem:MGF-Relations}, we get
\begin{align}
F_{A'}\left(s \right) & = F\left((s)_{A=A'}~,~(0)_{\substack{A\in\V':\\A\neq A'}}\right) \\
&= \sum_{\substack{(u,v)\in E':\\v\notin A'\text{ or }u\in A'}} \frac{\mu_{uv}}{\mu - s} F_{A'}\left(s \right) + \sum_{\substack{(u,v)\in E':\\v\in A'\text{ and }u\notin A'}} \frac{\mu_{uv}}{\mu - s} F_{A'\cup\{u\}}\left(s \right)\,.
\end{align} 
Rearranging the terms of the above equation yields~\eqref{eq:lem-MGF-marg-recur}.
\end{itemize}


\bibliographystyle{IEEEtran}
\bibliography{IEEEabrv,bibliography}

\end{document}